\documentclass[runningheads]{llncs}
\usepackage{microtype}
\usepackage{graphicx}
\usepackage{amsmath,dsfont}
\usepackage{subfigure}
\usepackage{scalerel}
\usepackage{todonotes}
\usepackage{thmtools}
\usepackage{thm-restate}

\usepackage{paralist}

\newcommand{\subcaptionhack}{
}

\newif\ifarxive
\arxivetrue %

\ifarxive
\newcommand{\springerarxiv}[2]{{#2}\xspace}
\else
\newcommand{\springerarxiv}[2]{{#1}\xspace}
\fi

\graphicspath{{figures/}}

\usepackage{hyperref}

\usepackage[nameinlink,capitalise]{cleveref}
\usepackage[utf8]{inputenc}
\hypersetup{colorlinks,linkcolor={blue},citecolor={blue},urlcolor={blue}} 

\usepackage{cleveref}
\Crefname{observation}{Observation}{Observations}
\Crefname{algorithm}{Algorithm}{Algorithms}
\Crefname{section}{Sect.}{Sects.}
\Crefname{observation}{Observation}{Observations}
\Crefname{lemma}{Lemma}{Lemmas}
\Crefname{claim}{Claim}{Claims}
\Crefname{figure}{Fig.}{Figs.}
\Crefname{figure}{Fig.}{Figs.}
\Crefname{enumi}{Property}{Properties}
\Crefname{property}{Property}{Properties}
\Crefname{remark}{Remark}{Remarks}

\usepackage{color}
\definecolor{realblue}{rgb}{0,0,1}
\definecolor{darkerblue}{rgb}{0.094,0.455,0.804}
\definecolor{darkblue}{rgb}{0.063,0.306,0.545}
\definecolor{red}{rgb}{0.627,0.117,0.156}
\definecolor{green}{rgb}{0,0.588,0.509}
\definecolor{orange}{rgb}{0.903,0.739,0.382}
\definecolor{realred}{rgb}{1,0,0}

\usepackage{xspace}

\newcommand{\blue}[1]{{{\textcolor{blue}{#1}\xspace}}}

\renewcommand{\emph}[1]{\blue{\em {#1}}}

\newcommand{\md}{\textup{mod}}

\DeclareMathOperator{\rot}{rot}
\DeclareMathOperator{\mi}{mid}

\title{Planar L-Drawings of Bimodal Graphs}

\author{Patrizio Angelini \inst{1}\orcidID{0000-0002-7602-1524}
\and
Steven~Chaplick \inst{2}\orcidID{0000-0003-3501-4608}
\and
Sabine~Cornelsen \inst{3}\orcidID{0000-0002-1688-394X}
\thanks{The work of Sabine~Cornelsen was funded by the
  German Research Foundation DFG – Project-ID 50974019 – TRR 161 (B06).}
\and 
Giordano~Da~Lozzo \inst{4}\orcidID{0000-0003-2396-5174}
\thanks{The work of Giordano~Da~Lozzo was partially supported by MIUR grant 20174LF3T8 {\em ``AHeAD: efficient Algorithms for HArnessing networked Data''}.}}
\authorrunning{P. Angelini et al.}
\institute{
  John Cabot University, Rome, Italy
  \href{mailto:pangelini@johncabot.edu}{pangelini@johncabot.edu}
  \and
  Maastricht University, The Netherlands
  \href{mailto:s.chaplick@maastrichtuniversity.nl}{s.chaplick@maastrichtuniversity.nl}
  \and
  University of Konstanz, Germany 
  \href{mailto:sabine.cornelsen@uni-konstanz.de}{sabine.cornelsen@uni-konstanz.de}
  \and
  Roma Tre University, Rome, Italy
  \href{mailto:giordano.dalozzo@uniroma3.it}{giordano.dalozzo@uniroma3.it}\\
}

\begin{document}
\color{black}

\maketitle

\begin{abstract}
  In a {\em planar L-drawing} of a directed graph (digraph) each~edge~$e$ is represented as a polyline composed of a vertical segment starting at the tail of $e$ and a horizontal segment ending at the head of $e$. Distinct edges may overlap,
  but not cross. Our main focus is on {\em bimodal graphs}, i.e.,
  digraphs admitting a planar embedding in which the
  incoming and outgoing edges around each vertex are contiguous.
  We show that every plane bimodal graph without 2-cycles
  admits a planar L-drawing. This includes the class of upward-plane graphs.
  Finally, outerplanar digraphs
  admit a planar L-drawing~--~although they do not always have a bimodal embedding~--~but not necessarily with an
  outerplanar embedding.
  \keywords{Planar L-Drawings \and Directed Graphs \and Bimodality}
\end{abstract}

\section{Introduction}

\begin{figure}[t]
$ $\hfill
  \subfigure[\label{SUBFIG:k4-2-modal}Bimodal]{\includegraphics[page=2]{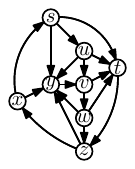}}\hfill
  \subfigure[\label{SUBFIG:k4-dual}Rectangular Dual]{\rule{10pt}{0em}\includegraphics[page=3]{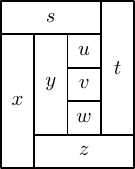}\rule{10pt}{0em}}\hfill $ $
  \subfigure[\label{SUBFIG:k4-Ldrawing}L-drawing]{\includegraphics[page=4]{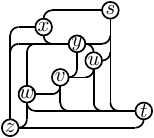}}\hfill
  \subfigure[\label{SUBFIG:k4-kandinsky}Kandinsky]{\includegraphics[page=5]{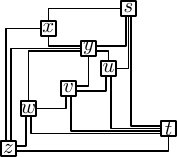}}\hfill $ $
  \caption{\label{FIG:k4}Various representations of a bimodal irriducible triangulation.}
\end{figure}

In an \emph{L-drawing} of a directed graph (digraph), vertices are
represented by points with distinct x- and y-coordinates, and each
directed edge $(u,v)$ is a polyline consisting of a vertical segment
incident to the tail $u$ and of a horizontal segment incident to the
head~$v$. Two edges may overlap in a subsegment with end point at a
common tail or head. An L-drawing is \emph{planar} if no two
edges cross (\cref{SUBFIG:k4-Ldrawing}).
Non-planar L-drawings were first defined by Angelini et
al.~\cite{DBLP:journals/ijfcs/KariOALBDPRT18a}. Chaplick et
al.~\cite{chaplick_etal:gd17} showed that it is NP-complete to decide
whether a directed graph has a planar
L-drawing if the embedding is not
fixed. However it can be decided in linear time whether a planar
st-graph has an \emph{upward-planar L-drawing}, i.e. an L-drawing in
which the vertical segment of each edge leaves its tail from the top.

A vertex $v$ of a plane digraph $G$ is
\emph{$k$-modal} ($\md(v) = k$) if in the cyclic sequence of edges
around $v$ there are exactly $k$ pairs of consecutive edges that are
neither both incoming nor both outgoing. 
A digraph $G$ is
\emph{$k$-modal} if $\md(v) \leq k$ for every vertex $v$ of~$G$. 
The 2-modal graphs are often referred to as \emph{bimodal}, see \cref{SUBFIG:k4-2-modal}.
Any plane digraph admitting a planar L-drawing is clearly  4-modal. 
Upward-planar and level-planar drawings induce bimodal embeddings. 
While testing whether a graph has a bimodal embedding is possible in linear time, 
testing whether a graph has a 4-modal embedding~\cite{DBLP:conf/esa/VialLG19} and testing whether a partial orientation of a plane graph can be extended to be bimodal~\cite{DBLP:journals/tcs/BinucciDP14} are NP-complete. 

A \emph{plane digraph} is a planar digraph with a fixed rotation
system of the edges around each vertex and a fixed outer face. In an
L-drawing of a plane digraph $G$ the clockwise cyclic order of the edges
incident to each vertex and the outer face is the one prescribed for
$G$.
In a planar L-drawing the edges 
attached to the same port 
of a vertex $v$ are ordered as follows: There are first the edges bending to the left with increasing length of the segment incident to $v$ and then those bending to the right with decreasing length of the segment incident to $v$.

This is analogous to the Kandinsky model~\cite{foessmeier/kaufmann:gd95} where vertices are
drawn as squares of equal size on a grid and edges as
orthogonal polylines on a finer grid (\cref{SUBFIG:k4-kandinsky}). Bend-minimization in the
Kandinsky model is NP-complete~\cite{blaesius_etal:esa14} and can be
approximated within a factor of two~\cite{barth_etal:gd06}.  Each
undirected simple graph admits a Kandinsky drawing with
one bend per edge~\cite{brueckner:BA}.  The relationship between Kandinsky drawings and planar L-drawings was established in~\cite{chaplick_etal:gd17}.

L-drawings of directed graphs can be considered as bend-optimal
drawings, since one bend per edge is necessary in order
to guarantee the property that edges must leave a vertex from the top
or the bottom and enter it from the right or the left.  
Planar L-drawings can be also seen as a directed version of
$\scalerel*{\textbf{+}}{\textsf{T}}$-contact representations, where
each vertex is drawn as a $\scalerel*{\textbf{+}}{\textsf{T}}$ and two
vertices are adjacent if the respective
$\scalerel*{\textbf{+}}{\textsf{T}}$es touch. 
If the graph is bimodal
then the
$\scalerel*{\textbf{+}}{\textsf{T}}$es are \textsf{T}s (including
\rotatebox[origin=c]{90}{\textsf{T}},
\rotatebox[origin=c]{180}{\textsf{T}}, and
\rotatebox[origin=c]{270}{\textsf{T}}). Undirected planar graphs
always allow a \textsf{T}-contact representation, which can be computed utilizing Schnyder
woods~\cite{fraysseix/mendez/rosenstiehl:94}. 

Biedl and Mondal~\cite{biedl/mondal:arxiv17} showed that a
$\scalerel*{\textbf{+}}{\textsf{T}}$-contact representation can also
be constructed from a rectangular dual (\cref{SUBFIG:k4-dual}). A
plane graph with four vertices on the outer face has a rectangular
dual if and only if it is an inner triangulation without separating
triangles~\cite{kozminski/kinnen:networks85}.  Bhasker and Sahni~\cite{bhasker/sahni:88} gave the first
linear time algorithm for computing rectangular duals.
He~\cite{he:93} showed how to compute a rectangular dual from a regular edge labeling and Kant and
He~\cite{kant/he:wg93} gave two linear time algorithms for computing
regular edge labelings.
Biedl and Derka~\cite{biedl/derka:jgaa16} computed
rectangular duals via (3,1)-canonical orderings.

\subsubsection{Contribution:}
We show that every bimodal graph without 2-cycles admits a planar
L-drawing respecting a given bimodal embedding. This implies that every upward-planar graph admits a planar L-drawing
respecting a given upward-planar embedding. We thus solve an open problem posed in~\cite{chaplick_etal:gd17}. The
construction is based on rectangular duals.
Finally, we show that every outerplanar graph admits a planar
L-drawing but not necessarily one where all vertices are incident to
the outer face. We conclude with open problems.

Proofs for statements marked with $(\star)$ can be found in\springerarxiv{ the full version~\cite{arxivVersion}}{ the appendix}, where we also provide an iterative algorithm showing that any bimodal graph with 2-cycles admits a planar L-drawing if the underlying undirected graph without 2-cycles is a planar 3-tree. 

\section{Preliminaries}

\subsubsection{L-Drawings.}
For each vertex we consider four \emph{ports}, North, South, East, and West.
An L-drawing implies a \emph{port assignment}, i.e.\ an assignment of the edges to the ports of the end vertices such that the outgoing edges are assigned to the North and South port and the incoming edges are assigned to the East and West port.  
A port assignment for each edge $e$ of a digraph $G$ defines a pair
(out$(e)$,in$(e))\in\{$North,South$\}\times\{$East,West$\}$.  
An L-drawing \emph{realizes} a port assignment
if each edge $e=(v,w)$ is incident to the out$(e)$-port of $v$
and to the in$(e)$-port of $w$. A port assignment \emph{admits} a planar
L-drawing if there is a planar L-drawing 
that realizes it. Given a port assignment it can be tested in linear
time whether it admits a planar L-drawing~\cite{chaplick_etal:gd17}.

In this paper, we will distinguish between given
L-drawings of a triangle.
\begin{restatable}[$\star$]{lemma}{Ltriangles}
  \label{LEMMA:Ltriangles}
    \cref{FIG:RDouterface} shows all planar L-drawings of~a~triangle up to symmetry.
\end{restatable}

\subsubsection{Coordinates for the Vertices.}

Given a port assignment that admits a planar L-drawing, a planar
L-drawing realizing it can be computed in linear time by the general
compaction approach for orthogonal or Kandinsky
drawings~\cite{eiglsperger/kaufmann:gd01}. However, in this approach,
the graph has to be first augmented such that each face has a
rectangular shape.
For
L-drawings of plane triangulations it suffices to make sure that each
edge has the right shape given by the port assignment, which can be
achieved using topological orderings only.

\begin{restatable}[$\star$]{theorem}{drawingTriangulation}\label{THEO:drawing-triangulations}
	Let $G=(V,E)$ be a plane triangulated graph with a port assignment
	that admits a planar L-drawing and let $X$ and $Y$ be the digraphs
	with vertex set $V$ and the following edges. For each edge
	$e=(v,w) \in E$
	\begin{itemize}
		\item
		there is $(v,w)$ in $X$ if
		in$(e)=$ West and $(w,v)$ in $X$ if in$(e)=$ East.
		\item there is $(v,w)$ in $Y$ if out$(e)=$
		North and $(w,v)$ in $Y$ if out$(e)=$ South.
	\end{itemize}
	Let $x$ and $y$ be a topological ordering of $X$ and $Y$,
    respectively. Drawing each vertex $v$ at $(x(v),y(v))$ yields a planar
    L-drawing realizing the given port assignment.
\end{restatable}

Observe that we can modify the edge lengths in a planar L-drawing
independently in x- and y-directions in an arbitrary way, as long as we
maintain the ordering of the vertices in x- and y-direction,
respectively. This will still yield a planar L-drawing. This fact
implies the following remark.
\begin{remark}\label{REM:givenOuterFace}
  Let $G$ be a plane digraph with a triangular outer face, let
  $\Gamma$ be a planar L-drawing of $G$, and let $\Gamma_0$ be a
  planar L-drawing of the outer face of $G$ such that the edges on the
  outer face have the same port assignment in $\Gamma$ and
  $\Gamma_0$. Then there exists a planar L-drawing of $G$ with the
  same port assignment as in $\Gamma$ in which the drawing of the
  outer face is $\Gamma_0$.
\end{remark}	

\subsubsection{Generalized Planar L-Drawings.}
An \emph{orthogonal polyline} $P=\left<p_1,\dots,p_n\right>$ is a sequence of points 
s.t.\ $\overline{p_ip_{i+1}}$ is vertical or horizontal.
For $1\leq i \leq n-1$ and a point $p\in \overline{p_ip_{i+1}}$, the polyline $\left<p_1,\dots,p_i,p\right>$ is a \emph{prefix} of $P$ and the polyline
$\left<p,p_{i+1},\dots,p_n\right>$ is a \emph{suffix} of $P$.  Walking from $p_1$ to
$p_n$, consider a \emph{bend} $p_i$, $i=2,\dots,n-1$.  The rotation
$\rot(p_i)$ is $1$ if $P$ has a left turn at $p_i$, $-1$ for a right turn, and $0$ otherwise (when $\overline{p_{i-1}p_i}$ and $\overline{p_ip_{i+1}}$ are both vertical or horizontal).
The \emph{rotation} of $P$ is $\rot(P)=\sum_{i=2}^{n-1}\rot(p_i)$.

In a \emph{generalized planar L-drawing} of a digraph, vertices are
still represented by points with distinct x- and y-coordinates and the
edges by orthogonal polylines with the following three properties. (1) Each
directed edge $e=(u,v)$ starts with a vertical segment incident to the
tail $u$ and ends with a horizontal segment incident to the head~$v$.
(2) The polylines representing two edges overlap in at most a common
straight-line prefix or suffix,
and they do not cross.

In order to define the third property, let init$(e)$ be the prefix of
$e$ overlapping with at least one other edge, let final$(e)$ be the
suffix of $e$ overlapping with at least one other edge, and let
mid$(e)$ be the remaining individual part of $e$. Observe that the
first and the last vertex of init$(e)$, final$(e)$, and mid$(e)$ are
end vertices of $e$, bends of $e$, or bends of some other edges. 
Now
we define the third property:
(3) For an edge $e$ one of the following
is true: (i) neither of the two end points of mid$(e)$ is a bend of $e$
and $\rot(e)=\pm 1$ or (ii) one of the two end points of mid$(e)$,
but not both, is a bend of $e$ and rot(mid$(e))=0$. See
\cref{FIG:generalized}.
As a consequence of the flow model of Tamassia~\cite{tamassia:87}, we
obtain the following lemma.

\begin{figure}[h]
\centering
$ $\hfill
\subfigure[\label{SUBFIG:generalized}not a generalized planar L-drawing]{\rule{2cm}{0cm}\includegraphics[page=1]{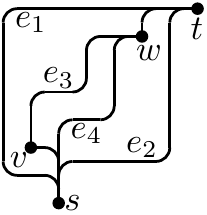}\rule{2cm}{0cm}}\hfill
\subcaptionhack
\subfigure[\label{SUBFIG:subdivided}underlying orthogonal drawing]{\rule{1.4cm}{0cm}\includegraphics[page=2]{generalized}\rule{1.4cm}{0cm}}
\hfill$ $
\caption{\label{FIG:generalized}Cond.~3 of generalized planar L-drawings is
  fulfilled for all edges but for $e_1$ and $e_2$. The rotation of
  each edge is $\pm1$. However, $\rot(\mi(e_1))=2$ and both end
  vertices of $\mi(e_2)$ are bends of $e_2$. }
\end{figure}

\begin{restatable}[$\star$]{lemma}{genL}
\label{LEMMA:genL}
  A plane digraph admits a planar L-drawing if and only if it admits a
  generalized planar L-drawing with the same port assignment. 
\end{restatable}

\subsubsection{Rectangular Dual.} An \emph{irreducible triangulation}
is an internally triangulated graph without separating triangles, where the outer face
has degree four (\cref{FIG:k4}).  A \emph{rectangular tiling} of a
rectangle $R$ is a partition of $R$ into a set of non-overlapping
rectangles such that no four rectangles meet at the same point. A
\emph{rectangular dual} of a planar graph is a rectangular tiling such
that there is a one-to-one correspondence between the inner rectangles
and the vertices and there is an edge between two vertices if and only
if the respective rectangles touch. We denote by $R_v$ the rectangle
representing the vertex $v$. Note that an irreducible triangulation
always admits a rectangular dual, which can be computed in linear
time~\cite{bhasker/sahni:88,biedl/derka:jgaa16,he:93,kant/he:wg93}.

\subsubsection{Perturbed Generalized Planar L-drawing.}
Consider a rectangular dual for a directed irreducible
triangulation $G$. We construct a drawing of $G$ as follows. We place
each vertex of $G$ on the center of its rectangle. Each edge is routed
as a \emph{perturbed orthogonal polyline}, i.e., a polyline within the
two rectangles corresponding to its two end vertices, such that each
edge segment is parallel to one of the two diagonals of the 
rectangle containing it. See \cref{SUBFIG:vw_can}.
This drawing is called a
\emph{perturbed generalized planar L-drawing} if and only if (1) each
directed edge $e=(u,v)$ starts with a segment on the diagonal
$\backslash_u$ of $R_u$ from the upper left to the lower right corner
and ends with a segment on the diagonal $\slash_v$ of $R_v$ from the
lower left to the upper right corner. Observe that a change of
directions at the intersection of $R_v$ and $R_u$ is not considered a
bend if the two incident segments in $R_v$ and $R_u$ are both parallel
to $\backslash$ or to $\slash$.  The definition of rotation and
Conditions~(2) and (3) are analogous to generalized planar L-drawings.

In a perturbed generalized planar L-drawing, the North port of a
vertex is at  the segment between the center and the upper left
corner of the rectangle. The other ports are defined analogously.
Since we can always approximate a segment with an orthogonal polyline
(\cref{SUBFIG:perturbed,SUBFIG:zigzag,SUBFIG:zigzagrotated}), we
obtain the following.

\begin{figure}[h]
\centering
\subfigure[\label{SUBFIG:vw_can}$(u,v)$]{\includegraphics[page=7]{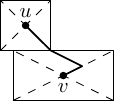}}
\hfill
\subfigure[\label{SUBFIG:perturbed}perturbed]{$\;\;$\includegraphics[page=3]{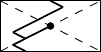}$\;\;$}
\hfill
\subcaptionhack
\subfigure[\label{SUBFIG:zigzag}zig zag]{\includegraphics[page=4]{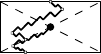}}
\hfill
\subcaptionhack
\subfigure[\label{SUBFIG:zigzagrotated}rotate]{\includegraphics[page=5]{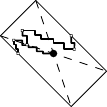}}
\hfill
\subfigure[\label{SUBFIG:zigzagsplit}split]{\includegraphics[page=9]{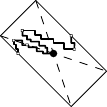}}
\caption{\label{FIG:zigzag} (a) An edge in a perturbed generalized planar L-drawing. (b-e) From a perturbed generalized planar L-drawing to a generalized planar L-drawing. }
\end{figure}
\begin{restatable}[$\star$]{lemma}{perturbed}
	\label{LEMMA:perturbed}
	If a directed irreducible triangulation has a perturbed generalized planar L-drawing, then it has a planar
	L-drawing with the same~port~assignment. 
\end{restatable}
\section{Planar L-Drawings of Bimodal Graphs}
We study planar L-drawings of plane bimodal graphs. 
Our main contribution is to show that if the graph does not contain any 2-cycles, then it admits 
a planar L-drawing (\cref{THEO:main}). \springerarxiv{In the full version of this paper~\cite{arxivVersion}}{In \cref{THEO:3-tree} in \cref{APP:3tree}}, we also show that if there are 2-cycles, then there is 
a planar L-drawing if the underlying undirected graph after removing parallel edges created by the 2-cycles is a planar 3-tree.

\subsection{Bimodal Graphs without 2-Cycles}
\label{SEC:2-modalfromRD}
  Our approach is inspired by the
work of Biedl and Mondal~\cite{biedl/mondal:arxiv17} that constructs a
$\scalerel*{\textbf{+}}{\textsf{T}}$-contact representation for
undirected graphs from a rectangular dual. 
We
extend their technique in order to respect the given orientations of
the edges.

The idea is to triangulate and decompose a given bimodal graph $G$. Proceeding from the outermost to the innermost
4-connected 
components, we construct planar L-drawings of
each component  that respects a given shape of the outer face.
We call a pair of edges $e_1,e_2$ a \emph{pincer} if $e_1$ and $e_2$ are on a triangle $T$, both are incoming or both outgoing edges of
its common end vertex $v$ (i.e.\ $v$ is a \emph{sink-} or a \emph{source switch} of $T$), and there is another edge $e$ of $G$ incident to
$v$ in the interior of $T$ but with the opposite direction. See \cref{FIG:badpincers}. 
If the outer face of a 4-connected component contains a pincer, we have to make sure that $e_1$ and $e_2$ are not assigned
to the same port of $v$ in an ancestral component. 
In a partial perturbed generalized planar L-drawing of $G$, we call a pincer \emph{bad} if $e_1$ and $e_2$ are assigned to the same port.  Observe that in a
bimodal graph, a pincer must be  a source or a sink in an ancestral component. Moreover, in a 4-connected component at most one
pair of incident edges of a vertex can be a pincer.

\begin{figure}[t]
\centering
    \subfigure[\label{SUBFIG:pincers}pincers]%
    {\includegraphics[page=6]{examples}}
\hfill
\hspace{-1cm}\subcaptionhack%
\subfigure[\label{SUBFIG:badpincers}bad pincers]%
    {\rule{10pt}{0pt}\includegraphics[page=7]{examples}\rule{10pt}{0pt}}
\hfill
\hspace{-1cm}\subcaptionhack%
\subfigure[\label{SUBFIG:nobadpincers}alternative]%
	{\rule{10pt}{0pt}\includegraphics[page=8]{examples}\rule{10pt}{0pt}}
    \hfill
    \hspace{-1cm}\subcaptionhack%
    \subfigure[\label{SUBFIG:K4pincer}2-modal pincer]%
    {\rule{10pt}{0pt}\includegraphics[page=1]{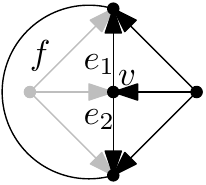}\rule{10pt}{0pt}}
    \hspace{-0.3cm}\subcaptionhack%
    \subfigure[\label{SUBFIG:virtualEdge}virtual edges]%
{\rule{10pt}{0pt}\includegraphics[page=2]{proofMainLemma}\rule{10pt}{0pt}}
\caption{\label{FIG:badpincers}(a) The blue edges incident to $v$ and $w$, respectively, are pincers that are bad in the drawing of the blue triangle in (b) and not bad in (c). (d) shows the only case (up to reversing directions) of a graph $H$ in \cref{SEC:4connected} with a pincer that is incident to a 2-modal vertex (the orientation of the undirected outer edge is irrelevant). (e) Avoiding bad pincers with virtual edges.}
\end{figure}

\begin{figure}[t]
\begin{center} \includegraphics[page=2,width=\textwidth]{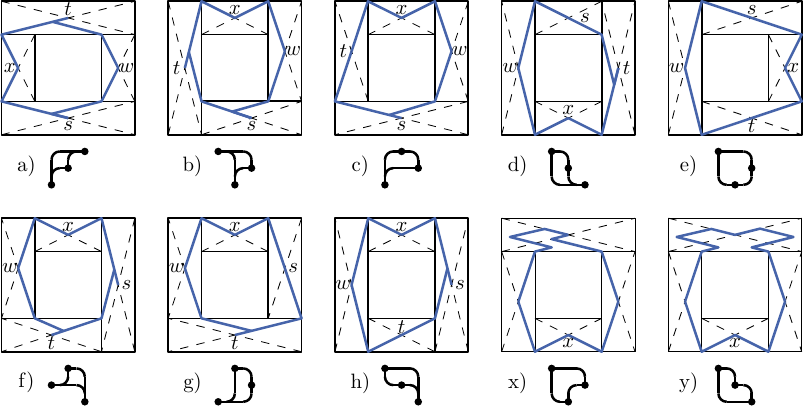} \end{center}
\caption{\label{FIG:RDouterface}Realization in the rectangular dual for any kind of drawings of the outer face up to symmetries.}
\end{figure}

\begin{theorem}\label{THEO:main}
  Every plane bimodal graph without 2-cycles admits a planar L-drawing. Moreover, such a drawing can be constructed in linear time.
\end{theorem}
\begin{proof}
	Triangulate the graph as follows: Add a new directed triangle in the outer face.
	Augment the graph by adding edges to obtain a plane bimodal graph in which each face has degree at most four
    as shown in\springerarxiv{ the full version~\cite{arxivVersion}}{ \cref{LEMMA:quadrangulate} in \cref{APP:lemmas}}.
    More precisely, now each non-triangular face is bounded by a 4-cycle consisting of alternating source and sink switches of the face. We finally
	insert a 4-modal vertex of degree 4 into each non-triangular face
   maintaining the 2-modality of the neighbors. Let $G$ be the obtained triangulated graph. We
   construct a port assignment that admits a
    planar L-drawing of $G$ as follows. Decompose $G$ at separating
  triangles into 4-connected components.  Proceeding from the outermost to the innermost components, we compute a port assignment for each 4-connected component $H$, avoiding bad pincers and such that the ports of the outer face of $H$ are determined by the corresponding inner face of the parent component of $H$. See \cref{SEC:4connected}.
  By \cref{THEO:drawing-triangulations}, we compute a planar L-drawing realizing the given port assignment. 
  Finally, we remove the added vertices and edges from $\Gamma$.  
  Since  the augmentation of $G$ and its decomposition into 4-connected components~\cite{DBLP:journals/ijcga/Kant97} can be performed in linear time, the total running~time~is~linear.\end{proof}

\noindent
\cref{THEO:main} yields the following
implication, solving an open problem in~\cite{chaplick_etal:gd17}.

  \begin{corollary}
	Every upward-plane graph admits a planar L-drawing.
\end{corollary}

\subsection{Planar L-Drawings for 4-Connected Bimodal Triangulations}\label{SEC:4connected}
In this subsection, we present the main algorithmic tool for the proof of \cref{THEO:main}. 
  Let $G$ be a triangulated plane digraph
without 2-cycles in which each vertex is 2-modal or an inner vertex of degree four. Let $H$ be a 4-connected component of $G$ (obtained by decomposing $G$ at its separating triangles) and 
let $\Gamma_0$ be a planar L-drawing of the outer face of $H$ without bad pincers of $G$. 
We now present an algorithm that constructs a planar L-drawing of $H$
in which the drawing of the outer face is $\Gamma_0$ and no face contains bad pincers~of~$G$. 

\subsubsection{Port Assignment Algorithm.
}
	The aim of the algorithm is to compute a port assignment for the edges of $H$ such that (i) there are no bad pincers and (ii) there exists a planar L-drawing realizing such an assignment. 
	Note that 
	the drawing $\Gamma_0$ already determines an assignment of the external edges to the ports of the external vertices. 
	By \cref{REM:givenOuterFace} any planar L-drawing with this given port assignment can be turned into one where the outer face has drawing $\Gamma_0$. 

	First, observe that $H$ does not contain vertices on the outer face that are 4-modal in $H$: %
	This is true since 4-modal vertices are inner vertices of degree four in the triangulated graph $G$ and since $G$ has no 2-cycles.
  This implies that $H$, likewise $G$, is a triangulated plane digraph without 2-cycles in which each vertex is 2-modal or an inner vertex of degree four.

   \paragraph{Avoiding Bad Pincers.}
	Next, we discuss the means that will allow us to avoid bad pincers.
	Let $e_1$ and $e_2$ be two edges with common end vertex $v$ that are incident to an inner face $f$ of $H$ such that $e_1,e_2$ is a pincer of $G$. Note that the triangle bounding $f$ is a separating triangle of $G$.  We call $f$ the \emph{designated face} of $v$. In the following we can assume that $v$ is 0-modal in $H$: In fact, if $v$ is 2-modal in $H$ then $v$ was an inner 4-modal vertex of degree 4 in $G$, and $e_1$ and $e_2$ are two non-consecutive edges incident to $v$. It follows that $H$ is a $K_4$ where the outer face is not a directed cycle. See \cref{SUBFIG:K4pincer}.
	For any given drawing $\Gamma_0$ of the outer face (see \cref{FIG:RDouterface} and \cref{LEMMA:Ltriangles} for the possible drawings of a triangle), the inner vertex can always be added such that no bad pincer is created. 
	Finally, observe that $v$ cannot be 4-modal in $H$ otherwise it would be at least 6-modal in $G$. 
	
	Hence, in the following, we only have to take care of pincers where the common end vertex is 0-modal in $H$. 
  Since each 0-modal vertex was 2-modal in $G$, it 
  has at most one designated face. In the following, we assume that all 0-modal vertices are assigned a designated incident inner face where no $0^\circ$~angle~is~allowed.

    \paragraph{Constructing the Rectangular Dual.}
  As an intermediate step towards a perturbed generalized planar L-drawing, we have to construct a rectangular dual of $H$, more precisely of an irreducible triangulation obtained from $H$ as follows.
  Let $s$, $t$, and $w$ be the vertices on the outer face of $H$.  Depending
  on the given drawing of the outer face, subdivide one of the edges
  of the outer face by a new vertex $x$ according to the cases given
  in \cref{FIG:RDouterface}~--~up to symmetries. Let $f$ be the
  inner face incident to $x$.  Then $f$ is a quadrangle. Triangulate
  $f$ by adding an edge $e$ incident to $x$: Let $y$ be the other end
  vertex of $e$. If $y$ was 2-modal, we can orient~$e$ such that
  $y$ is still 2-modal. If $y$ was 0-modal and $f$ was its
  designated face, then orient $e$ such that $y$ is now
  2-modal. Otherwise, orient $e$ such that $y$ remains
  0-modal. Observe that if $y$ had degree 4 in the
  beginning it has now degree 5.

  The resulting graph $H_x$ is triangulated, has no separating
  triangles and the outer face is bounded by a quadrangle, hence it is
  an irreducible triangulation. Thus, we can compute a rectangular
  dual $R$ for $H_x$.  Up to a possible rotation of a multiple of
  $90^\circ$, we can replace the four rectangles on the outer face
  with the configuration depicted in \cref{FIG:RDouterface} that
  corresponds to the given drawing of the outer face.  Let $R_v$ be
  the rectangle of a vertex $v$.

  \paragraph{Port Assignment.}
  We now assign edges to the ports of the incident vertices.  For the
  edges on the outer face the port assignment is given by $\Gamma_0$.
  \cref{FIG:RDouterface} shows the assignments for the outer face.

  Let $v$ be a vertex of $H_x$.  We define the \emph{canonical
    assignment} of an edge incident to $v$ to a port around $v$ as
  follows (see \cref{SUBFIG:canonical}). An outgoing edge $(v,u)$ is
  assigned to the North port, if $R_u$ is to the left or the top of
  $R_v$. Otherwise it is assigned to the South port. An incoming edge
  $(u,v)$ is assigned to the West port, if $R_u$ is to the left or the
  bottom of $R_v$. Otherwise it is assigned to the East port.
  
  In the following we will assign the edges to the ports of their end
  vertices such that each edge is assigned in a canonical way to at
  least one of its end points and such that crossings between edges
  incident to the same vertex can be avoided within the rectangle of
  the common end vertex. We exploit this property alongside with the
  absence of 2-cycles to prove that such an assignment determines a
  perturbed generalized planar L-drawing of the plane graph $H_x$.

  \begin{figure}[t]
    \centering
    \subfigure[\label{SUBFIG:canonical}canonical assignment]{\includegraphics[page=35]{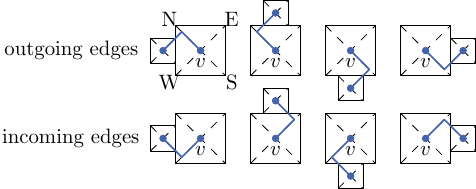}}\hfill
    \subcaptionhack
    \subfigure[\label{SUBFIG:outerx}around $y$]{\includegraphics[page=9]{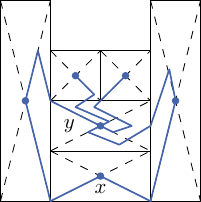}}\hfill    
    \caption{\label{FIG:canonical}Port assignment: (b) around neighbor $y$ of outer subdivision vertex $x$.}
  \end{figure}

  \paragraph{0-Modal Vertices.}    
  We  consider each 0-modal vertex $v$ to be 2-modal by adding a \emph{virtual
  edge} inside its designated face $f$. Namely, suppose $v$ is a source and let $e_1=(v,w_1)$ and
  $e_2=(v,w_2)$ be incident to $f$. We add a virtual
  edge $(w,v)$ between $e_1$ and $e_2$ from a new \emph{virtual vertex}
  $w$.  Of course, there is not literally a rectangle $R_w$
  representing $w$, but for the assignments of the edges to the ports
  of $v$, we assume that $R_w$ is the degenerate rectangle corresponding to the segment on the
  intersection of $R_{w_1}$ and $R_{w_2}$. See \cref{SUBFIG:virtualEdge}.

  \paragraph{2-Modal Vertices.}
  Let now $v$ be a 2-modal vertex.  We discuss the cases where we have to
  deviate from the canonical assignment. We call a side $s$ of a
  rectangle in the rectangular dual to be \emph{mono-directed},
  \emph{bi-directed}, or \emph{3-directed}, respectively, if there are
  0, 1, or 2 changes of directions of the edges across $s$. See \cref{FIG:assignmentPorts}. Observe
  that by 2-modality there cannot be more than two changes of
  directions.

  Consider first the case that $R_v$ has a side that is 3-directed, say the right side of
  $R_v$. See \cref{SUBFIG:3-directed}. If from top to bottom there are first outgoing edges followed by
  incoming edges and followed again by outgoing edges, then we assign from top to
  bottom first the North port, then the East port, and then the
  South port to the edges incident to rectangles on the right of $R_v$
  (\emph{counterclockwise switch}). Otherwise, we
  assign from top to bottom first the East port, then the South port,
  and then the West port (\emph{clockwise
    switch}). All other edges are assigned in a canonical way to the ports of $v$; observe
  that there is no~other~change~of~directions.

  \begin{figure}[t]
    \centering
    $\quad\quad$
    \subfigure[\label{SUBFIG:3-directed}3-directed]{\includegraphics[page=36]{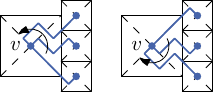}}\hfil
    \subcaptionhack
    \subfigure[\label{SUBFIG:canonicalBi}canonical]{$\;\;$\includegraphics[page=41]{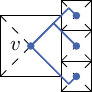}$\;\;$}\hfil  
    \subcaptionhack
    \subfigure[\label{SUBFIG:clockwise}clockwise]{$\;\;$\includegraphics[page=43]{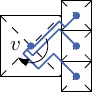}$\;\;$}\hfil
    \subcaptionhack
    \subfigure[\label{SUBFIG:counterc}counterclockw.]{$\quad\;\;\;$\includegraphics[page=42]{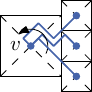}$\quad\;\;\;$}\\
    \hfil
    \subfigure[\label{SUBFIG:dependenciesSides}dependencies]{\includegraphics[page=38]{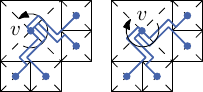}}\hfil
    \subcaptionhack
    \subfigure[\label{SUBFIG:monoDirected}mono-directed]{\includegraphics[page=39]{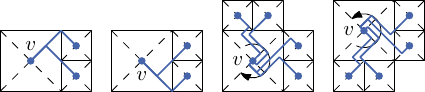}}\hfil
    \caption{\label{FIG:assignmentPorts}Assignment of ports when the direction of edges incident to one side of a rectangle changes a) twice b-e) once, or f) never.} %
  \end{figure}

  Consider now the case that $R_v$ has one side that is bi-directed, say again the right side of $R_v$. If the order from top to
  bottom is first incoming then outgoing then assign the edges
  incident to the right side of $R_v$ in a canonical way (\emph{canonical switch}, \cref{SUBFIG:canonicalBi}). Otherwise \emph{(unpleasant switch)},  we
  have two options, we either assign the outgoing edges to the
  North port and the incoming edges to the East port
  (\emph{counter-clockwise switch}, \cref{SUBFIG:counterc}) or we assign the outgoing edges to
  the South port and the incoming edges to the West port
  (\emph{clockwise switch}, \cref{SUBFIG:clockwise}).

  Observe that if there is an unpleasant
  switch on one side of $R_v$ then
  there cannot be a canonical switch
  on an adjacent side. Assume now that there are two adjacent
  sides $s_1$ and $s_2$ of $R_v$ in this clockwise order around $R_v$ with
  unpleasant switches.
  Then we consider both switches
  as counterclockwise or both as clockwise.
  See \cref{SUBFIG:dependenciesSides}. 
  Observe that due to 2-modality two opposite sides of $R_v$ are neither
  both involved in unpleasant switches nor both in canonical switches.

  Consider now the case that one side $s$ of $R_v$ is mono-directed, say again the right side of $R_v$. See \cref{SUBFIG:monoDirected}. In most cases, we
  assign the edges incident to $s$ in a canonical way. There would be~--~up
  to symmetry~--~the following exceptions: The top side of $R_v$ was involved in a
  clockwise switch and the edges at the right side are incoming
  edges. In that case we have a \emph{clockwise switch} at $s$, i.e.,
  the edges at the right side are assigned to the West port of
  $v$. The bottom side of $R_v$ was involved in a counter-clockwise switch and the edges
  at the right side are outgoing edges. In that case we have a
  \emph{counter-clockwise switch} at $s$, i.e., the edges at the right
  side are assigned to the North port of $v$. 

  In order to avoid switches at mono-directed sides, we do the following:
  Let $s_1,s,s_2$ be three consecutive sides 
  in this clockwise order around the rectangle $R_v$ such that there is
  an unpleasant switch on side $s$; say $s$ is the right side of
  $R_v$, $s_1$ is the top and $s_2$ is the bottom, and the edges on
  the right side are from top to bottom first outgoing and then
  incoming. By 2-modality, there cannot be a switch of directions on
  both, $s_1$ and $s_2$, i.e., $s_1$ contains no incoming edges, or $s_2$
  contains no outgoing edges. In the first case, we opt for a
  counterclockwise switch for $s$, otherwise, we opt for a clockwise switch.

  \begin{figure}
  	\centering
    \includegraphics[page=16]{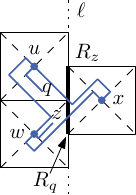}  
\caption{\label{SUBFIG:designatedcorner}\sc Extra Rule}
    \end{figure} 

  There is one exception to the rule in the previous paragraph (which we refer to as {\sc Extra Rule}): Let $u$ and $w$ be two adjacent
  0-modal vertices with the same designated face $f$ such that
   the two virtual end vertices are on one line
  $\ell$. See \cref{SUBFIG:designatedcorner}. Let $s_u$ be the side of $R_u$ intersecting $R_w$ and let
  $s_w$ be the side of $R_w$ intersecting $R_u$. Assume that $u$ has
  an unpleasant switch at $s_u$ and, consequently, $w$ has an
  unpleasant switch at $s_w$.  Let $x$ be the third vertex on $f$ and
  let $s_x$ be the side of $R_x$ intersecting $R_u$ and $R_w$. Observe
  that $s_x \subset \ell$.  Do the switch at $s_u$ and $s_w$ in clockwise
  direction if and only if the switch at $s_x$ is in clockwise
  direction, otherwise in counterclockwise direction.

  \begin{property}\label{PROP:mono-directed}
    There is neither a clockwise nor a
    counter-clockwise switch at a mono-directed side of a rectangle
    $R_v$ except if $v$ is one of the 0-modal vertices to which the
    {\sc Extra Rule} was applied.
  \end{property}

  \paragraph{4-Modal Vertices.}
  If $v$ is an inner 4-modal vertex of degree 4, then each side of
  $R_v$ is incident to exactly one rectangle, and we always
  use the canonical assignment. If $v$ is a
  4-modal vertex of degree 5, then $v$ is the inner
  vertex $y$ adjacent to the subdivision vertex $x$. Note that we do
  not have to draw the edge between $x$ and $y$. However, this case is
  still different from the previous one, since there are two rectangles 
  incident to the same side $s$ of
  $R_y$. If the switch at $s$ is canonical then there is no
  problem. Otherwise we do the assignment as~in~\cref{SUBFIG:outerx}.

  Observe that we get one edge between $y$ and a vertex on the outer
  face that is not assigned in a canonical way at $y$. But this edge
  is assigned in a canonical way at the vertex in the outer face. This
  completes the port assignments.

  \subsubsection{Correctness.}
  In \springerarxiv{the full version~\cite{arxivVersion}}{\cref{APP:correctness}}, we give a detailed proof that the constructed port
  assignment admits a perturbed generalized planar L-drawing and, thus, a planar
  L-drawing of $H$. The proof starts with the observation that each
  edge is assigned in a canonical way at one end vertex at least. Then
  we route the edges as indicated in \cref{FIG:routing}
  where each
  part of a segment that is not on a diagonal of a rectangle
  represents a perturbed orthogonal polyline of rotation 0. Finally,
  we show that the encircled bends are not contained in any other
  edge.
\begin{figure}[h!]
	\centering
	\includegraphics[page=13,width=\textwidth]{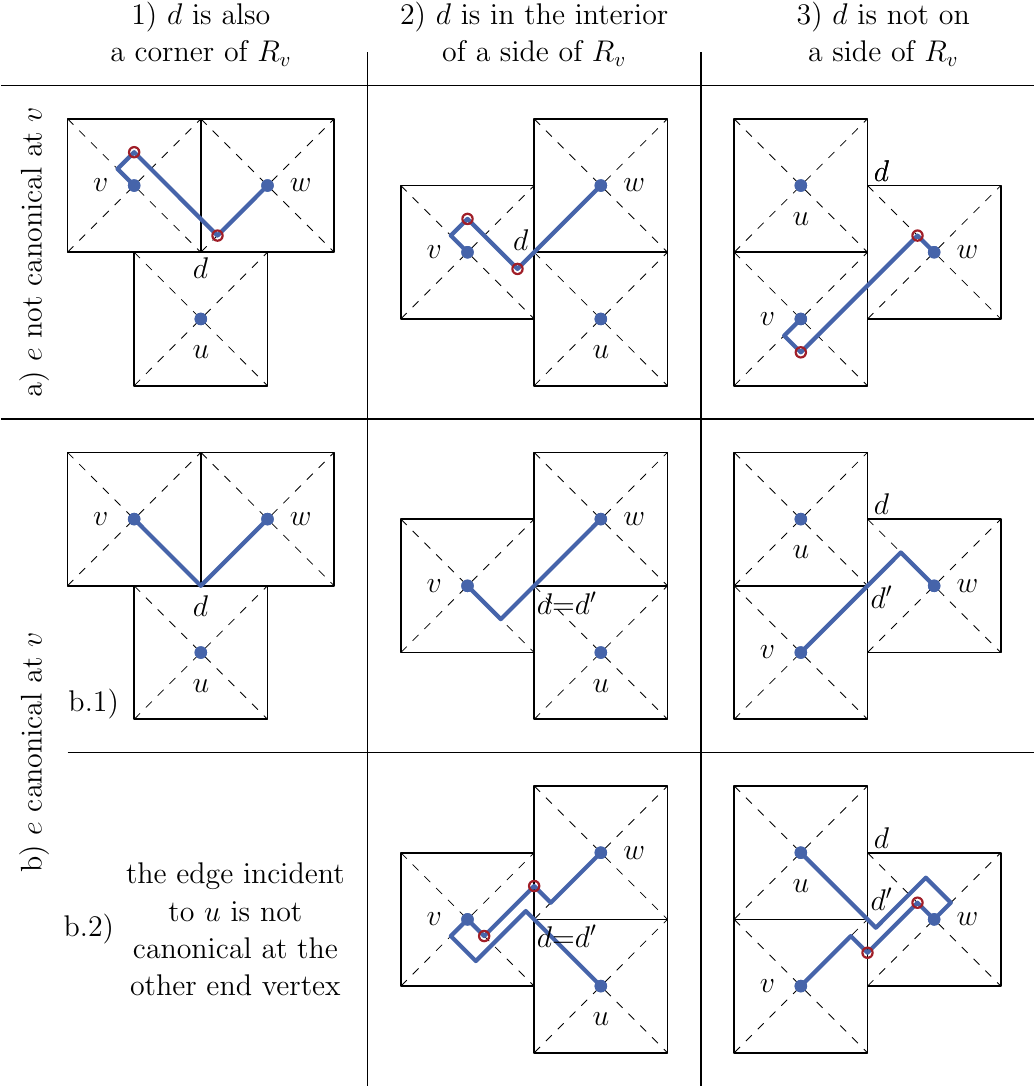}
	\caption{\label{FIG:routing}How to route the edge $e$ between $v$ and $w$. Point $d$ is the corner at the end of the diagonal of $R_w$ to which $e$ is assigned.}
\end{figure}

{
  \begin{lemma}\label{LEMMA:linearTime}	
	A planar L-drawing of $H$ in which the drawing of the outer face
    is $\Gamma_0$ and no face contains bad pincers of $G$ can be
    constructed in linear time.
\end{lemma}
\begin{proof}
  The construction
  guarantees a planar L-drawing of $H$. The
  port assignment is such that there are no bad pincers and for the
  outer face it is the same as in $\Gamma_0$. A rectangular dual can
  be constructed in linear
  time~\cite{bhasker/sahni:88,kant/he:wg93}. The port assignment can
  also be done in linear time. Finally, the coordinates can be
  computed in linear time using topological
  ordering~--~see~\cref{THEO:drawing-triangulations}.
\end{proof}
}

\section{Outerplanar Digraphs}

Since
there exist outerplanar digraphs that do not admit any bimodal
embedding~\cite{DBLP:conf/esa/VialLG19}, we cannot exploit
\cref{THEO:main} to construct planar L-drawings for the graphs in this
class. However, we are able to prove the following.

\begin{theorem}\label{lem:outerplanar-l-planar}
  Every outerplanar graph admits a planar L-drawing.
\end{theorem}

\begin{proof}
  Put all vertices on a diagonal in the order in which they appear on
  the outer face~--~starting from an arbitrary vertex. The drawing of
  the edges is determined by the direction of the edges. This
  implies that some edges are drawn above and some below the
  diagonal. By outerplanarity, there are no
  crossings.%
\end{proof}

We remark that \cref{lem:outerplanar-l-planar} provides an alternative
proof to the one in~\cite{DBLP:conf/esa/VialLG19} that any outerplanar
digraph admits a $4$-modal embedding. Observe that the planar L-drawings constructed in the proof of \cref{lem:outerplanar-l-planar} are not necessarily outerplanar. In the following, we prove that this may be unavoidable.

\begin{theorem}
	Not every outerplanar graph admits an outerplanar L-drawing.
\end{theorem}

\centerline{\includegraphics[page=33]{figures}}

\begin{proof}
  Consider the graph depicted above. It has a unique outerplanar
  embedding. Let $f$ be the inner face of degree 6.  Each vertex incident
  to $f$ is 4-modal and is a source switch or a sink switch of
  $f$. Thus, the angle at each vertex is $0^\circ$. The angle at each
  bend is at most $3/2 \pi$. Thus, the angular sum around $f$ would
  imply $ (2\cdot\deg f-2)\cdot \pi \leq 3/2 \cdot \deg f \cdot \pi, $
  which is not possible for $\deg f = 6$.
\end{proof}

There are even 4-modal biconnected internally triangulated
outerplane digraphs that do not admit an outerplanar L-drawing. See \springerarxiv{the full version~\cite{arxivVersion}}{\cref{APP:outer}}.

\section{Open Problems}

\begin{itemize}
\item Are there bimodal graphs with 2-cycles that do not admit a
  planar L-drawing (with or without the given embedding)?
\item What is the complexity of testing whether a 4-modal graph admits
  a planar L-drawing with a fixed embedding?
\item In the directed Kandinsky model where edges leave a vertex to
  the top or the bottom and enter a vertex from the left or the right,
  for which $k$ is there always a drawing with at most $1+2k$ bends
  per edge for any 4-modal graph? $k=0$ does not suffice. What about
  $k=1$?
\item Can it be tested efficiently whether an outerplanar graph with a
  given 4-modal outerplanar embedding admits an outerplanar L-drawing?
\end{itemize}

\bibliographystyle{splncs04}
\bibliography{main}

\springerarxiv{\end{document}}{}
\appendix

\chapter*{Appendix}

\section{Proof of Some Lemmas}\label{APP:lemmas}

\begin{lemma}
	\label{LEMMA:quadrangulate}
	Let $G$ be a plane digraph. We can augment $G$ by
	adding edges to obtain a biconnected plane digraph with
	face-degree at most four such that each $k$-modal vertex remains
	$k$-modal if $k>0$ and each 0-modal vertex gets at most 2-modal.
	The construction neither introduces parallel edges nor 2-cycles.
    Moreover, each 4-cycle bounding a face consists of alternating source and sink switches.
\end{lemma}

\begin{proof}
  Let $G=(V,E)$ be a plane directed graph.
\begin{enumerate}
\item If $G$ is not connected, let $G_1$ be a connected component of
  $G$ such that the outer face $f_o^1$ of $G_1$ (considered as an open
  region) contains vertices of $G$.  Pick a vertex $v$ incident to
  $f_o^1$ with the property that $v$ is an isolated vertex or $v$ is
  the tail of an edge incident to $f_o^1$.  Let $f$ be the face of $G$
  that is contained in $f_o^1$ and that is incident to $v$. Pick a
  vertex $w$ incident to $f$ that is not in $G_1$ and such that $w$ is
  isolated or the head of an edge incident to $f$. Add the edge
  $(v,w)$.
\item If $G$ contains a cut vertex $v$, let $w_1$ and $w_2$ be two
  consecutive neighbors of $v$ in different biconnected
  components. Let $f$ be the face between the edges connecting $v$ to $w_1$ and $w_2$, respectively. If
  we can add the edge $(w_1,w_2)$ or $(w_2,w_1)$ such that the
  modalities of $w_1$ and $w_2$ do not increase, if they had been
  positive before, we do so. Otherwise we may assume that the edges
  incident to $f$ and $w_1$ and $w_2$ are all outgoing edges of $w_1$
  and $w_2$, respectively, and that the degree of $w_2$ is at least
  two. Let $w_3 \neq v$ be a neighbor of $w_2$ on $f$. Add $(w_1,w_3)$
  to $G$.
\item If $G$ contains a face $f$ of degree greater than four, let
  $v_1,\dots,v_k$ be the facial cycle of $f$ such that $v_1$ is the
  tail of an edge incident to $f$.  Let $i = 3,4$ be minimum such that
  $v_i$ is 0-modal or incident to an incoming edge on $f$. If neither
  $(v_1,v_i)$ nor $(v_i,v_1)$ is present in $G$ then add $(v_1,v_i)$
  to $G$. Otherwise we can add an edge between $v_2$ and $v_{4}$ or
  $v_{5}$ if $i=3$ or between $v_5$ and $v_3$ or $v_2$ if $i=4$.
\end{enumerate}
\end{proof}

\Ltriangles*
\begin{proof}
  We first consider a triangle $T$ that is not a directed cycle. Let
  $t$ be the sink, $s$ the source, and $w$ the third vertex of $T$.
  We may assume that $\left<s,w,t\right>$ is the counter-clockwise
  cyclic order of vertices around $T$ and that the edge $(s,w)$ uses
  the North port of $s$~--~the other cases being symmetric. We
  distinguish two cases.
  \begin{enumerate}
  \item $w$ is to the right of $s$. In this case $t$ cannot be below
    $w$: otherwise it is not possible to close $T$ in
    counter-clockwise direction with only one bend per edge. For the
    x-coordinate of $t$ there are three possibilities: $t$ is to the
    right of $w$ (a), between $s$ and $w$ (c), and to the left of $s$ (b).
  \item $w$ is to the left of $s$. Similar as in the first case, $t$
    cannot be above $w$ in this case. If $t$ is below $s$, its
    x-coordinate can be to the left of $w$ (g), between $w$ and $s$ (e), or to
    the right of $s$ (d). If $t$ is vertically between $w$ and $s$, its
    x-coordinate is between $w$ and $s$ (h) or to the left of $w$ (f).
  \end{enumerate}
  Consider now the case that triangle $T$ is a directed cycle. Then no
  two edges of $T$ can use the same port. Thus, $T$ is drawn as a
  6-gone. I.e., the angular sum is $4\pi$ which implies that there is
  one $3\pi/2$ angle and five $\pi/2$ angles. We distinguish the case
  where the $3\pi/2$ angle is at a vertex (y) or at a bend (x).
\end{proof}

\genL*
\begin{proof}
  A planar L-drawing is a planar generalized L-drawing. So assume that
  a planar generalized L-drawing $\Gamma$ of a digraph $G$ is
  given. We consider $\Gamma$ as an orthogonal drawing of a new graph
  $G'$~--~see \cref{SUBFIG:subdivided}. To this end, we replace
  every bend in $\Gamma$ that is contained in the polyline
  representing another edge by a dummy vertex. Consider now an edge of
  $G$ that has more than one bend. This edge is decomposed in $G'$
  into an initial straight-line path $P_\text{init}$, an edge $e_\text{mid}$, and a
  final straight-line path $P_\text{final}$. The edge
  $e_\text{mid}$ is represented by an orthogonal polyline with the same
  bends as mid$(e)$.  Using the flow model of
  Tamassia~\cite{tamassia:87}, we can remove
  all but rot(mid($e$))
  bends from $e_\text{mid}$. Now the cases are two: If none of the end
  vertices of $e_\text{mid}$ is a bend of $e$ then $\rot(e)=\pm 1$
  and, thus, $e$ ends up with exactly one bend. If the first or the
  last bend of $e$ is an end vertex of $e_\text{mid}$ then
  $\rot(e_\text{mid})=0$ and $e_\text{mid}$ ends up with no
  bends. Thus, $e$ has exactly one bend at exactly one end vertex of
  $e_\text{mid}$.
\end{proof}

\perturbed*
\begin{proof}
	Let a perturbed generalized planar L-drawing $\Gamma$ of an irreducible triangulation $G$ be
	given.  By \cref{LEMMA:genL}, it suffices to show that $G$ has
	a generalized planar L-drawing. First, we construct a graph $G'$ of
	maximum degree 4 by subdividing the edges at all bends and
	intersections with the boundary of a rectangle. Observe that now
	every edge lies in the interior of one rectangle. We approximate
	each edge $e$ of $G'$ arbitrarily close with an orthogonal polyline
	rotated by $45^\circ$ in such a way that the following properties
	hold: No two polylines of two edges cross and the polyline of an
	edge is not self-intersecting. The rotation of any polyline is
	zero. The segments incident to the end vertices have both slope
	$-45^\circ$ if $e$ is parallel to $\backslash_v$ and slope
	$45^\circ$ if $e$ is parallel to $\slash_v$, where $R_v$ is the
	rectangle containing $e$. See \cref{FIG:zigzag}.
	Finally, rotate
	the drawing by $45^\circ$ in clockwise direction.  We obtain a drawing that fulfills all properties of a 
	generalized planar L-drawing of $G$, except that edges might overlap in a prefix or a suffix with rotation 0 that might, however, not be a straight-line segment.
	We fix this as follows. Let $v$ be a vertex and let $\left<e_1,\dots,e_k\right>$ be the sequence of edges assigned to a port of $v$  in clockwise order. Assume without loss of generality that $e_1,\dots,e_k$ are outgoing edges of $v$. Let $1 \leq m \leq k$ be such that init$(e_m)$ is longest. We redraw each edge $e_i$ with $i \neq m$ so that the common part of $e_i$ and $e_m$ lies on the first segment and the lengths of init$(e_1),\dots,$init$(e_{m-1})$ are still increasing and those of init$(e_{m+1}),\dots,$init$(e_{k})$ are still decreasing. The rest of $e_i$ is drawn arbitrarily close to $e_m$ in such a way that the rotation of mid$(e_i)$ is maintained. See \cref{SUBFIG:zigzagsplit}.
	  Observe that the bends of the original drawing still correspond to
	bends of the constructed drawing and have the same turns (left
	or right).  Conditions~2+3 are still fulfilled.
\end{proof}

\section{Proof of \cref{THEO:drawing-triangulations}}\label{APP:coordinates}
\drawingTriangulation*
We use the following lemma in order to prove \cref{THEO:drawing-triangulations}. 
\begin{lemma}\label{LEMMA:planarIfFaces}
	A drawing of a plane biconnected graph is planar if the ordering of
	the edges around each vertex is respected and the boundary of each
	face is drawn crossing-free.
\end{lemma}
\begin{proof}
  The lemma can be proven by the same proof idea as in the proof of
  Tutte~--~see Item~9.1 on page 758 of \cite{tutte63}.  Let $G$ be a
  plane biconnected graph and let $\Gamma$ be a drawing of $G$ in
  which each face is drawn crossing-free.  Assume two edges $e_1$ and
  $e_2$ cross. Consider two faces $f_1$ and $f_2$ whose boundary
  contains $e_1$ and $e_2$, respectively. (We consider the boundary of
  a face not to be part of the face. In particular are faces open
  regions). Then there must be a point $q \in f_1 \cap f_2$. But this
  is impossible: For each point $p$ in the plane let $\delta(p)$ be
  the number of faces containing $p$. Since the inner faces are
  bounded, there must be a point $q_o$ that is only contained in the
  outer face and thus, $\delta(q_o)=1$. Consider a curve $\ell$ from
  $q_o$ to $q$ that does not contain vertices. Traversing $\ell$, the
  count $\delta$ does not change if no edge is crossed. If we cross an
  edge then we leave a face and enter another face. Thus, the count
  will not change either, contradicting $\delta(q)\geq 2$.
\end{proof}

\begin{figure}
	\centering
	$ $\hfill
	\subfigure[]{\includegraphics[page=17]{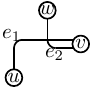}}
	\hfill
	\subcaptionhack\subfigure[]{\includegraphics[page=18]{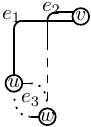}}
	\hfill
	\subcaptionhack\subfigure[]{\includegraphics[page=13]{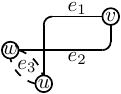}}
	\hfill $ $
	\caption{\label{FIG:drawing-triangulation}A triangular face cannot self-intersect if all edges have the correct shape.}
\end{figure}

Consider two edges $e_1$ and $e_2$ in a planar L-drawing that are
assigned to the same port of a vertex $v$ and let $f$ be the face
between $e_1$ and $e_2$. Then the bend at $e_1$ or $e_2$ must be
concave in $f$. This is a subcondition of the so called
bend-or-end property of Kandinsky drawings and we refer to it as the
\emph{concave-bend condition}.

\begin{proof}[of \cref{THEO:drawing-triangulations}]
	It suffices to show that all faces are drawn in a planar
	way. Assume there are two edges $e_1$ and $e_2$ incident to the same face that cross.  
	Since all faces are triangles $e_1$ and $e_2$ are incident. Assume without loss of
	generality that out$(e_1)=$ North, in$(e_1)=$ West, and that the
	head $v$ of $e_1$ is an end vertex of $e_2$. We distinguish three
	ways $e_2$ could cross $e_1$ (see
	\cref{FIG:drawing-triangulation}): 
	\begin{enumerate}[(a)]
	\item 
	  $v$ is the head of $e_2$,
	in$(e_2)=$ West, out$(e_2)=$ South, and $e_2$ is before $e_1$ in
	the clockwise order around $v$, or 
	\item 
	$v$ is the head of $e_2$,
	in$(e_2)=$ West, out$(e_2)=$ North, and $e_2$ is after $e_1$ in the
	clockwise order around $v$, or 
	\item 
	$v$ is the tail of $e_2$,
	out$(e_2)=$ South,  in$(e_2)=$ East,  the head $w$ of $e_2$ is
	to the left, and above of the tail $u$ of $e_1$.
    \end{enumerate} 
    Situation~(a)
	violates the concave-bend condition. A crossing in Situation~(b)
	implies that the edge $e_3$ closing the triangular face is either
	$(u,w)$ with in$(e_3)=$ West or $(w,u)$ with in$(e_3)=$ East. 
	However, this port assignment would not be one that admits a planar L-drawing. Finally, a crossing in Situation~(c)  implies that the edge $e_3$
	closing the triangular face is either $(u,w)$ with in$(e_3)=$ East
	and out$(e_3)=$ North or $(w,u)$ with in$(e_3)=$ West and out$(e_3)=$
	South. Again, this port assignment would not be one that admits a planar L-drawing. 
\end{proof}

Observe that in general it does not suffice to only consider the
right drawing of each edge in order to obtain a planar L-drawing even if the port assignment admits such a drawing, not even for a
directed 4-cycle (\cref{FIG:4-cycle}) or a 4-cycle consisting only of sink and source switches
(\cref{FIG:4-cycle-noinout}).

\begin{figure}
  \centering
  $ $\hfill
  \subfigure[\label{FIG:4-cycle}directed 4-cycle]{\includegraphics[page=10]{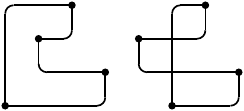}}
  \hfill
  \subcaptionhack\subfigure[\label{FIG:4-cycle-noinout}4-cycle]{\includegraphics[page=11]{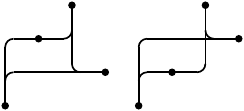}}
  \hfill
  \caption{Given a port assignment that admits a planar L-drawing, not
    every L-drawing that realizes it is also planar. }
\end{figure}

\section{Correctness of the Port Assignment in \cref{SEC:4connected}}\label{APP:correctness}
  
  \begin{lemma}\label{LEMMA:canonical} 
  	Any edge $e=(u,w)$ is assigned in a canonical way at $u$ or $w$.
  \end{lemma}
  \begin{proof}
  Suppose, for contradiction, that there is an edge $e=(u,w)$ that is neither assigned in a canonical way at $u$ nor at
  $w$.
  Assume, without loss of generality, that $R_u$ is on top of $R_w$. This implies that $e$ is assigned to the North port of $u$ and to the
  West port of $w$. Moreover, the bottom side $s_u$ of $R_u$ is involved in a
  clockwise switch and the top side $s_w$ of $R_w$ in a counter-clockwise
  switch. See \cref{FIG:canonical}.

  \begin{figure}[h!]
  \centering
  \includegraphics[page=28]{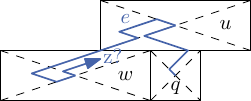}
  \caption{\label{FIG:canonical}Edges are canonical at one end vertex at least.}
  \end{figure}

  First consider the case that neither $s_u$ nor $s_w$ is mono-directed. 
  Then the only reason for these unpleasant switches is 
  that 
  \begin{inparaenum} 
  \item $u$ has an incoming edge $(q,u)$ and $R_q$ is
  incident to the bottom of $R_u$ and is to the right of $R_w$, and 
  \item  $w$ has an outgoing edge $(w,z)$ such that $R_z$ is
  incident to the top of $R_w$ and to the right of $R_u$.
  \end{inparaenum} 
  This is not
  simultaneously possible if at least one among $R_z$ or $R_q$ is a real rectangle.
  In the case that $R_z$ and $R_q$ were both
  virtual, the designated  face of $u$ and $w$ would be the same, $R_z$ and $R_q$ would be collinear and, thus, the
  {\sc Extra Rule} would apply. However, this implies that the switches at $u$ and $w$ are
  both clockwise or both counterclockwise. See \cref{SUBFIG:designatedcorner}.

  Consider now the case that one of the two sides $s_u$ and $s_w$, say $s_w$, is mono-directed. By \cref{PROP:mono-directed}, a non-canonical switch at a mono-directed side can only happen in the case of the {\sc Extra Rule}. See \cref{FIG:canonicalextrarule}.

  \begin{figure}[h!]
  \centering
  \includegraphics[page=19]{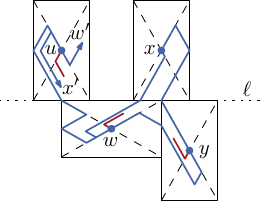}%
  \caption{\label{FIG:canonicalextrarule}Illustration of the proof of \cref{LEMMA:canonical} in the case of the {\sc Extra Rule}.}
  \end{figure}

  We distinguish two cases based on whether $s_u$ is bidirected or mono-directed.
  Suppose first that $s_u$ is bidirected. As above there must still be a neighbor $q$ of $u$ such that $R_q$ is to the right of $R_w$ and below $R_u$. But this is not possible.
  
  Assume now that also $s_u$ is mono-directed, i.e. the {\sc Extra Rule} was also applied to $u$ and a vertex $w'$ and the designated face $f'$ of $u$ and $w'$ is as indicated by the red stub incident to $u$ in \cref{FIG:canonicalextrarule}. Let $x'$ be the third vertex incident to $f'$. Observe that $R_u$ and $R_{w'}$ must have a common corner that lies on a side of $R_{x'}$. However, this would only be possible if the bottom right corner of $R_u$ and the bottom left corner of $R_{w'}$ lie on the top side of $R_{x'}$ which implies that $x' = w$. But $w$ is a 0-modal vertex and $x'$ cannot be 0-modal. 
  \end{proof}

  \begin{lemma}\label{LEMMA:construction}	
  	The constructed port assignment admits
  	 a planar L-drawing of $H$. 
  \end{lemma}
  \begin{proof}
  In the following we show how to route the edges in order to obtain a perturbed generalized planar L-drawing with the given port assignment. By \cref{LEMMA:perturbed} this is sufficient to obtain a planar L-drawing with the same port assignment.

  Recall that in a perturbed
  generalized planar L-drawing each edge is routed
  as a polyline lying in the two rectangles corresponding to its end vertices, composed of segments parallel to the respective diagonals and satisfying the following properties:  
  \begin{enumerate}[(1)]
  	\item\label{ITEM:diagonale} each directed edge
  $e=(u,v)$ starts with a segment on the diagonal $\backslash_u$ of
  $R_u$ from the upper left to the lower right corner and ends with a
  segment on the diagonal $\slash_v$ of $R_v$ from the lower left to the
  upper right corner. 
  \item\label{ITEM:common} The polylines representing two edges overlap in at most a common
  straight-line prefix or a suffix
  and they do not cross.
  \item\label{ITEM:rot} For an edge $e$ one of the following
  is true: (i) none of the two end vertices of mid$(e)$ is a bend of $e$
  and $\rot(e)=\pm 1$ or (ii) one of the two end vertices of mid$(e)$,
  but not both, is a bend of $e$ and rot(mid$(e))=0$.
  \end{enumerate}

  \cref{ITEM:diagonale} is already fulfilled by the constructed port assignment. Next, we show that the port assignment allows for a routing of the edges that also fulfills \cref{ITEM:common,ITEM:rot}.

  Let $e$ be an edge between $v$ and $w$ that is drawn in a canonical
  way at $w$. Let $e$ be assigned to the port $p$ of $R_w$. 
  Let $d$ be the corner of $R_w$ at the end of the diagonal corresponding to the port $p$.
  We
  define how to route $e$ distinguishing three main cases on the
  relationship of $R_v$ and $R_w$ with respect to $d$~--~see the columns of \cref{FIG:routing}. Each case has two subcases
  depending on whether $e$ is assigned in a canonical way at $v$ or
  not~--~see the rows of \cref{FIG:routing}. 
    We subdivide the case where $e$ is assigned in a canonical way at
  $v$ into two additional subcases b.1 and b.2. Let $u$ be the common neighbor
  of $v$ and $w$ that is incident to the side of $w$ containing
  the corner $d$. Let $d'$ be the common point of $R_u$, $R_v$, and $R_w$.
  The subcases depend on whether an edge incident to $u$ would use $d'$ or not.   
    Observe that in the last column of Row~b.2 the corner $d$ does not have to be on a side of $R_u$ since $R_u$ might be
    smaller.
  
    For each of the cases, we draw $e$ as sketched in the corresponding box in \cref{FIG:routing}: The first and the last segment of an edge represents a straight line of an appropriate length while
    each other segment represents a perturbed orthogonal polyline with rotation zero. 
    If $e$ is assigned to the same port as another edge $e'$, we make sure that the routing respects the embedding by appropriately selecting the length of the first or last segment of $e$ and $e'$. E.g., assume that $e'$ follows  $e$ in counterclockwise order around a vertex $v$ such that both are assigned to the North port of $v$ and assume that the first bend of both, $e$ and $e'$, is a right turn.  Then the bend of $e$ is closer to $v$ than the bend of $e'$.
    
    Observe that the edges are only routed within the rectangles of their end vertices. Thus, if there were crossings then they would involve edges incident to the same vertex $v$ and would lie within the rectangle $R_v$. However,
    the port assignment at $v$ makes sure that the middle part of the edges can be routed such that no two edges cross. This guarantees that \cref{ITEM:common} is fulfilled.

  It remains to show that also \cref{ITEM:rot} is fulfilled. 
  This is trivial for the cases in Row~b.1, since in this case the edge is composed of two parts, each having rotation zero, that meet at a bend $b$. Thus, mid$(e)$ either starts at $b$ and has rotation 0, or it contains $b$ as an inner point and has rotation $\pm 1$. Analogously,  for the other cases, it suffices to prove that 
  mid$(e)$ contains two out of the three indicated bends as inner points~--~one with a left turn and one with a right turn.
  In fact, in this case mid$(e)$ has rotation 0 or $\pm1$ depending on whether the third bend is an end point of mid$(e)$ or not.
  
  To this end, we prove that each of the bends that are encircled red are inner points of mid$(e)$. 
   This is obvious, if the bend is not on a diagonal, since the end points of mid$(e)$ lie on the diagonals of the rectangles of the end vertices of $e$. If the bend is on the diagonal, we
  prove that there are no edges leaving the diagonal after the bend. If there was such an edge then it would be one that immediately follows or precedes $e$ in the cyclic order around the respective end vertex.
  \begin{description}
  \item[b.2)] 
  	We argue about the bend $b$ on the diagonal of $R_v$ in Column~2, the arguments for the bend on the diagonal of $R_w$ in Column~3 is analogous.  
  	Observe that the edge $e$ is immediately followed by the edge $(u,v)$ in the cyclic order around $v$. Since $(u,v)$ is assigned to a different port of $v$ than $e$, the statement follows. 
  	  \begin{figure}[h]
  		\centering
  		\subfigure[\label{SUBFIG:1a}Case~1a]{\includegraphics[page=14]{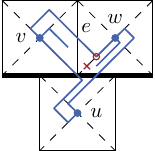}}\hspace{2cm}
  		\subcaptionhack\subfigure[\label{SUBFIG:1a}Case~3a]{\includegraphics[page=15]{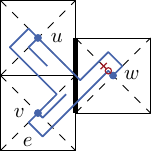}}
  		\caption{\label{FIG:routingSpecial}The red bend is an inner point of mid$(e)$ in the case of the  \mbox{\sc Extra Rule}.}
  	\end{figure}
  	  \item[1a+3a)] 
   These two
    cases cannot happen except if we had applied the {\sc Extra Rule}.  See  \cref{FIG:routingSpecial}.
    I.e., two among $u,v,w$ are 0-modal
    vertices, their designated face is the face bounded by $u,v,w$, the respective virtual rectangles are collinear, and there is
    an unpleasant switch at the third vertex. In both cases the non-virtual edge that immediately follows (3a) or precedes (1a) $e$ in the cyclic order around $w$ is the edge $(u,w)$. However, the port assignment in the  {\sc Extra Rule} 
    guarantees that $(u,w)$ is assigned to a different port of $w$ than $e$.
  \end{description}

This concludes the proof of the lemma.
\end{proof}

\section{Outerplanar Digraphs}\label{APP:outer}

\begin{theorem}\label{THEO:outerinner}
	Not every biconnected internally triangulated outerplanar digraph with
	a 4-modal outerplanar embedding has an outerplanar L-drawing.
\end{theorem}

We start the proof of the theorem with the following observation.

\begin{figure}
	\centering
	\includegraphics[page=2]{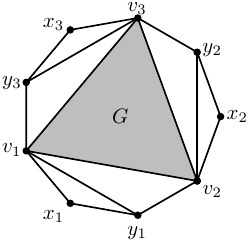}
	\caption{\label{FIG:outerExtension}How to make the vertices of a biconnected internally triangulated outerplanar digraph 4-modal.}
\end{figure}

\begin{lemma}\label{REM:outerExtension}
	Every biconnected internally triangulated outerplanar digraph $G$ with
	a 4-modal outerplanar embedding can be extended to an internally
	triangulated outerplanar digraph $G'$ with a 4-modal outerplanar
	embedding in which all vertices of $G$ are 4-modal. 
\end{lemma}
\begin{proof}
  See \cref{FIG:outerExtension} for an illustration. Let
  $v_1,\dots,v_n$ be the vertices of $G$ in the order in which they
  appear on the outer face and let $v_{n+1}=v_1$. For $i=1,\dots,n$
  add a new vertex $x_i$ with neighbor $v_i$ and a vertex $y_i$ with
  neighbors $x_i$, $v_i$, and $v_{i+1}$. Now each new vertex has
  degree at most three and thus, will be 2-modal, no matter how we
  orient the edges. Each vertex $v_i$, $i=1,\dots,n$ of $G$ is
  incident to three new edges. These can be oriented such that each
  $v_i$, $i=1,\dots,n$ gets 4-modal.
  \end{proof}

  \begin{figure}
	\centering
	\subfigure[\label{FIG:intTriangOut}graph]{\includegraphics[page=1]{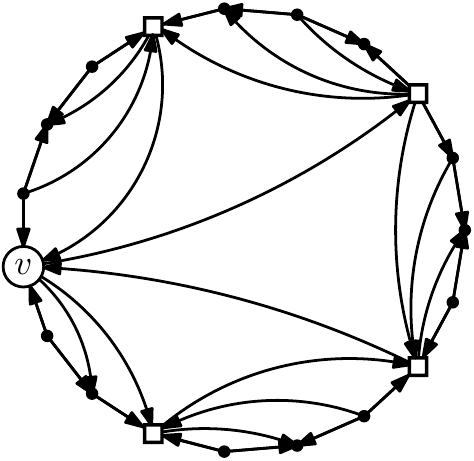}}
	\\
	\subfigure[\label{FIG:intTriangOutL}non-planar L-drawing]{\includegraphics[page=1]{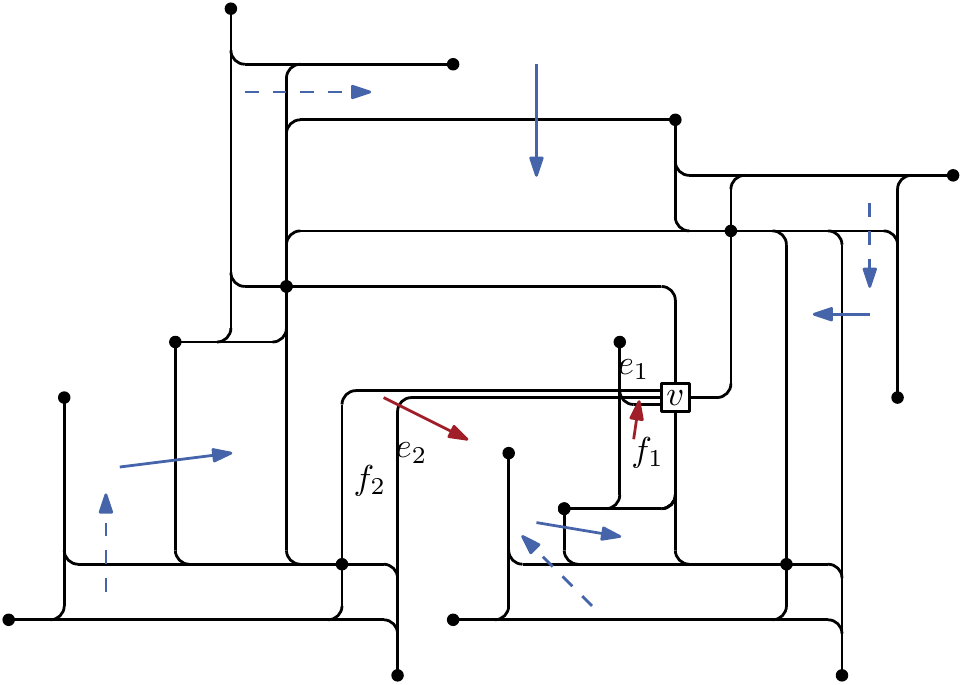}}
	\caption{A biconnected internally triangulated outerplanar digraph
      with a 4-modal outerplanar embedding that has no outerplanar
      L-drawing~--~in an extension that makes the present vertices
      4-modal. b) Different from the drawing conventions of
      L-drawings, the edges at the West port of $v$ are drawn slightly
      apart in order to make the embedding visible.  }
  \end{figure}

  \begin{proof}[of \cref{THEO:outerinner}]
  We consider the digraph $G=(V,E)$ in \cref{FIG:intTriangOut} augmented
  as described in \cref{REM:outerExtension}. The resulting digraph $G'$
  is a biconnected internally triangulated outerplanar digraph with a
  4-modal outerplanar embedding and has 57 vertices. Each vertex of
  $G$ is 4-modal in $G'$. We show that $G'$ has no planar L-drawing
  with the given embedding.  \cref{FIG:intTriangOutL} shows an
  L-drawing of $G$ that is, however, not planar.

  Consider now the following flow network $\mathcal L_G$ associated
  with a plane digraph $G=(V,E)$.  Let $F$ be the set of faces of $G$.
  $\mathcal L_G$ has node set $W:=V \cup E \cup F$, arcs from vertices
  to incident faces, and from faces to incident edges, and supplies
  $b(v) = \frac{4-\md(v)}{2}$, $v \in V$,
  $b(f) = \mp 2 + \# \textup{source switches of }f$, $f \in F$, and
  $b(e)= -1$, $e \in E$.

  Based on the relationship to Kandinsky drawings, Chaplick et
  al.~\cite{chaplick_etal:gd17} showed that any planar L-drawing of a
  biconnected digraph $G$ yields a flow $\phi$ in $\mathcal L_G$ as
  follows: Let $\alpha(v,f)$ be the angle in the face $f$ at a vertex
  $v$. Then $\phi(v,f) = \alpha(v,f)/\pi$ if the two edges incident to
  both, $v$ and $f$, are two outgoing or two incoming edges of $v$ and
  $\phi(v,f) = \alpha(v,f)/\pi -1/2$, otherwise. $\phi(f,e)=1$ if and
  only if there is a convex bend in face~$f$ on edge~$e$ and 0
  otherwise.

  However, not every flow on $\mathcal L_G$ corresponds to a planar
  L-drawing. In effect, the angles and bends in
  the non-planar L-drawing in \cref{FIG:intTriangOutL} also yield a
  flow $\phi'$ in $\mathcal L_G$. 
  
  Assume now that there is a planar L-drawing $\Gamma'$ of $G'$ and
  consider the drawing $\Gamma$ of $G$ induced by $\Gamma'$. Let
  $\phi$ be the flow that corresponds $\Gamma$.  The difference
  between $\phi$ and $\phi'$ is a union $C$ of directed cycles in the
  residual network $\mathcal L_{G,\phi'}$ of $\mathcal L_G$ with
  respect to the flow $\phi'$. Since in $G'$ all vertices of $G$ are
  4-modal, no arc from a vertex of $G$ to an inner face of $G$ can carry
  flow. Hence, $C$
  does not contain vertices of $G$. The direction of the arcs in
  $\mathcal L_{G,\phi'}$ is from a face $f$ to an incident edge $e$
  with a concave bend in $e$ and from $e$ to the other incident face
  $f'$.
	
  Since $\phi$ corresponds to a planar L-drawing it follows that at
  least one of the edges $e_1$ or $e_2$ must have a concave bend in
  the outer face $f_o$. Thus, $C$ contains the arcs
  $(f_2,e_2),(e_2,f_o)$ or the arcs $(f_1,e_1),(e_1,f_o)$. This
  implies that $C$ must also contain an arc from the outer face. But
  there are only three such arcs (solid blue arcs). Whenever such an
  arc is contained in $C$ then also the respective dashed blue arc has
  to be contained in $C$. Since in a union of directed cycles the
  indegree and the outdegree must be the same, $C$ can neither contain
  the arc $(e_2,f_o)$ nor $(e_1,f_o)$.
\end{proof}

\section{Planar 3-Trees}\label{APP:3tree}

A \emph{planar 3-tree} is defined recursively: The complete graph
$K_4$ on four vertices is a planar 3-tree. Adding a new vertex into an
inner face $f$ of a planar 3-tree and connecting it to the 3 vertices
on the boundary of $f$ yields again a planar 3-tree.
\begin{theorem}\label{THEO:3-tree}
  A bimodal graph has a planar L-drawing if the underlying undirected
  graph, after the removal of parallel edges due to 2-cycles, is a
  planar~3-tree.
\end{theorem}
\begin{proof}
  Observe that there are no separating 2-cycles in the digraph since
  planar 3-trees are 3-connected. Moreover, due to bimodality no
  vertex is incident to more than two 2-cycles.
	
  We start with a graph containing three vertices.  We draw that
  triangle with its 2-cycles and then keep inserting the vertices
  maintaining the invariant that there are no bad pincers. Observe
  that each inserted vertex has three adjacent vertices in the current
  digraph and up to five incident edges. We call a vertex an
  in/out-vertex of a face $f$ if it is neither a source switch nor a
  sink switch of $f$.
  
  \paragraph{So consider a bimodal graph with three vertices containing a triangle.} 
  If the outer face is a 2-cycle, remove one of the parallel edges,
  draw the triangle such that the respective vertices of the 2-cycle
  are extreme points of the diagonal of the bounding box of the
  drawing (\cref{FIG:sourceSink,FIG:IOSink,FIG:IOSource,FIG:IOIO}) and
  add the missing edge of the outer 2-cycle. If the outer 2-cycle had
  been between the source and the sink of the triangle then~--~due to
  bimodality~--~the triangle is not involved in any further
  2-cycles. Otherwise there can be at most one more 2-cycle. Bad
  pincers are not possible in these cases.  \newlength{\ruler}
  \setlength{\ruler}{0.7cm}
  \begin{figure}
    \subfigure[\label{FIG:sourceSink}sink--source]{\rule{\ruler}{0cm}\includegraphics[page=5]{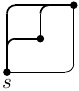}\rule{\ruler}{0cm}}
    \hfill
    \subcaptionhack\subfigure[\label{FIG:IOSink}sink--in/out]{\rule{\ruler}{0cm}\includegraphics[page=6]{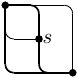}\rule{\ruler}{0cm}}
    \hfill
    \subcaptionhack\subfigure[\label{FIG:IOSource}in/out--source]{\rule{\ruler}{0cm}\includegraphics[page=7]{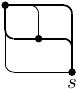}\rule{\ruler}{0cm}}
    \hfill
    \subcaptionhack\subfigure[\label{FIG:IOIO}in/out--in/out]{\rule{\ruler}{0cm}\includegraphics[page=8]{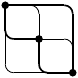}\rule{\ruler}{0cm}}\\
     \centering
     \subfigure[\label{FIG:directed}directed]{\rule{\ruler}{0cm}\includegraphics[page=9]{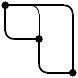}\rule{\ruler}{0cm}}
    \hspace{2cm}
    \subcaptionhack\subfigure[\label{FIG:notdirected}not directed]{\rule{\ruler}{0cm}\includegraphics[page=10]{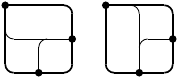}\rule{\ruler}{0cm}}
    
    \caption{\label{FIG:outer_face}Different ways of drawing a
      triangle $T$ with possible 2-cycles. Top row: the outer face is
      bounded by a 2-cycle where the edge not in $T$ connects the
      indicated vertices of $T$. Bottom row: the outer face is bounded
      by $T$ and $T$ is directed or not. }
  \end{figure}

  Consider now the case that the outer face is bounded by the triangle.
  We say that a \emph{triangle contains a 2-cycle} if one edge of the
  2-cycle is an edge of the triangle and the other edge of the 2-cycle
  is in the interior of the triangle.  Observe that in a bimodal
  embedding a triangle can contain at most two 2-cycles. Moreover, if
  a triangle contains two 2-cycles then the common vertex of the two
  2-cycles must be the source or the sink of the triangle. Thus, if
  the triangle is a directed cycle, it contains at most one 2-cycle.

  Now, if the outer face is a directed cycle of length three, draw it
  as indicated in \cref{FIG:directed}.  Otherwise,
  the three vertices incident to the outer face are a source, a sink,
  and an in/out vertex.  Start with the drawing of the triangle where
  no two edges are attached to the same port of a vertex. 
  Add the
  2-cycles. Thus, the outer face with its 2-cycles is a subgraph of
  one of the two cases in \cref{FIG:notdirected}. 
  
  Observe that in any case no two edges $e_1$ and $e_2$ that are
  assigned to the same port of a vertex $v$ can be a pincer.  There is
  already an edge in opposite direction incident to $v$. If there was
  an edge that had to be inserted later on between $e_1$ and $e_2$ and
  that also had the opposite direction as $e_1$ and $e_2$ then $v$
  would not be 2-modal.
  
  \paragraph{Assume now that we insert the next vertex $v$ with three
    neighbors.} If the neighbors of $v$ induce a directed triangle
  $T$, then there cannot be a pincer that consists of an edge in $T$
  and an edge incident to $v$: there is already an oppositely directed
  edge in $T$ incident to the common end vertex. For the same reason,
  if $v$ is neither a source nor a sink, then there cannot be a pincer
  incident to $v$.

  \begin{figure}
    \centering
    \includegraphics[page=1]{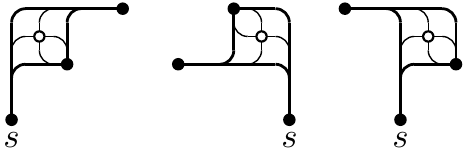}
    \caption{\label{FIG:upwardTriangle2pincer}Different ways of
      drawing a triangle where the angle at both, the source and the sink, is $0^{\circ}$
      and how to add an additional internal vertex.}
  \end{figure}

  Assume first that $T$ has a source $s$, a sink $t$ and an
  in/out-vertex $w$. There are three cases: (1) The angles in $T$ at
  $s$ and $t$, respectively, are both $0^{\circ}$
  (\cref{FIG:upwardTriangle2pincer}), (2) the angle at either the
  source or the sink~--~say the source~--~is $0^{\circ}$
  (\cref{FIG:vertexInTriangleAll}, Columns~d+e), or (3) $T$ does not
  contain a $0^{\circ}$ angle (\cref{FIG:vertexInTriangleAll},
  Column~c).
  \begin{figure}
    \centering
    \includegraphics[page=11]{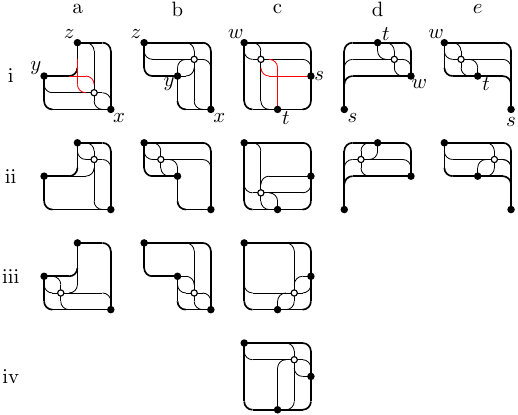}
    \caption{\label{FIG:vertexInTriangleAll}Different ways of inserting
      a vertex into a triangle.}
  \end{figure}
  
  In the first case~--~since there are no bad pincers~--~the direction
  of the edges between $v$ on one hand and the source and the sink on
  the other hand are fixed and it follows that~--~due to
  2-modality~--~$v$ cannot be incident to a 2-cycle. The pair of edges
  with a $0^\circ$ angle at the in/out-vertex $w$ of $T$ cannot be a
  pincer, since there is already an oppositely directed edge incident
  to $w$. This concludes Case~1.
  
  In Cases~2 and~3, we have the property that if there is a 2-cycle
  between $v$ and $w$ then~--~due to 2-modality of $w$~--~the order
  around $w$ is fixed and must be such that on both sides an edge of
  $T$ and an edge of the 2-cycle forms a $0^{\circ}$ angle.  Consider
  now Case~2 and assume w.l.o.g.\ that there are two edges of $T$ that
  form a $0^{\circ}$ angle at $s$. Up to symmetry there are two such
  drawings of $T$. See \cref{FIG:vertexInTriangleAll}, Columns~d+e. If
  $t$ is incident to a 2-cycle with $v$ then the order of the edges in
  the 2-cycle determines whether we are in case (i) or (ii). No two
  edges that are assigned to the same port are a pincer, since there
  is always an oppositely directed incident edge. If $t$ is not
  incident to a 2-cycle with $v$, we can choose between case (i) and
  (ii) and we do it such that we do not create a bad pincer incident
  to $t$.
  
  Consider now Case~3. There can be at most two 2-cycles incident to
  $v$, and if so, fixing the ordering of the edges of one 2-cycle also
  fixes the ordering of the other (due to 2-modality of
  $v$). Depending on this ordering, we can always choose one out of
  the Cases~ii or~iv. No bad pincers can occur. If $v$ is incident to
  at most one 2-cycle, the choice depends on pincers at $v$, $s$ and
  $t$ and can always without creating bad pincers: There are all four
  variants of pairs of $0^{\circ}$ angles at $s$ and $t$, so we can
  always choose one that does not contain a bad pincer. If $v$ is a
  source or a sink then $v$ is not incident to any 2-cycles. It
  follows that $T$ cannot contain edges incident to both $s$ and $t$
  that are involved in pincers with edges incident to $v$. E.g.,
  assume that $v$ is a source. Then only $(v,t)$ can be involved in a
  bad pincer: $(v,s)$ and an edge on $T$ incident to $s$ do not have
  the same direction at $s$, and $w$ is already incident to an edge in
  the opposite direction as $(v,w)$. Among the four possibilities we
  can always choose one without bad pincers.

  Assume now that $T$ is a directed triangle
  (\cref{FIG:vertexInTriangleAll}, Columns~a+b). Observe that due to
  2-modality $v$ cannot be incident to two 2-cycles. It remains to
  check for bad pincers at $v$ in the case were $v$ is a source or a
  sink. But there are always enough choices so they can be avoided.
\end{proof}

\end{document}